\documentclass{article}
\pdfoutput=1
\usepackage{graphicx}
\usepackage{amsmath,amssymb,verbatim,amsthm}
\newtheorem{theorem}{Theorem}[section]

\newtheorem{corollary}[theorem]{Corollary}

\newtheorem{remark}[theorem]{Remark}

\newtheorem{assumption}[theorem]{Assumption}

\numberwithin{equation}{section}

\usepackage{color}
\usepackage{xcolor}

\newcommand{\bbE}{\mathbb{E}}
\newcommand{\bbP}{\mathbb{P}}
\newcommand{\bbR}{\mathbb{R}}

\newcommand{\cG}{\mathcal{G}}
\newcommand{\cV}{\mathcal{V}}
\newcommand{\cB}{\mathcal{B}}

\newcommand{\testfunction}{\ensuremath{\mathcal{B}_b}}

\newcommand{\cO}{\mathcal{O}}
\newcommand{\cL}{\mathcal{L}}

\newcounter{hypA}
\newenvironment{hypA}{\refstepcounter{hypA}\begin{itemize}
  \item[({\bf A\arabic{hypA}})]}{\end{itemize}}

\title{Forward and Inverse Uncertainty Quantification using Multilevel Monte Carlo
Algorithms for 
an Elliptic Nonlocal Equation}

\author{
A. Jasra, \thanks{
Department of Statistics \& Applied Probability National University of Singapore Singapore, Singapore
  }
  \and
K.J.H. Law, \thanks{
  Computer Science and Mathematics Division, Oak Ridge National Laboratory, 
  Oak Ridge, TN, USA, 37831
  }
\and
Y. Zhou \thanks{
  Department of Statistics \& Applied Probability National University of Singapore Singapore, Singapore
  }
}

\begin{document}
\maketitle

\begin{abstract}
This paper considers uncertainty quantification for an elliptic nonlocal equation.  In particular, it is assumed that the parameters which define the kernel in the nonlocal operator are uncertain and a priori distributed according to a probability measure.  It is shown that the induced probability measure on some quantities of interest arising from functionals of the solution to the equation with random inputs is well-defined; as is the posterior distribution on parameters given observations.  
As the elliptic nonlocal equation cannot be solved approximate posteriors are constructed.
The multilevel Monte Carlo (MLMC) and 
multilevel sequential Monte Carlo (MLSMC) sampling algorithms are used for a priori and a posteriori 
estimation, respectively, of quantities of interest.  These algorithms 
reduce the amount of work to estimate posterior expectations, for a given level
of error, relative to Monte Carlo and i.i.d.~sampling from the posterior
at a given level of approximation of the solution of the elliptic nonlocal equation.\\
   
   \bigskip
   
\end{abstract}

\section{Introduction}
\label{sec:intro}


{\it Anomalous diffusion}, where the associated underlying stochastic process 
is not Brownian motion, 
has recently attracted considerable attention 
\cite{bakunin2008turbulence}. 
This case is interesting because 
there may be long-range correlations, among other reasons.  
Anomalous superdiffusion can be be related to fractional Laplacian operators
and/or so-called {\it nonlocal} operators \cite{du2012analysis}, defined point-wise  
by their operation on a function $u:\bbR^d \rightarrow \bbR$ as 
\begin{equation}
\cL(u)(x) := \int_{\bbR^d} B(x,x') [u(x') - u(x)]  dx'.  
\label{eq:ell}
\end{equation}
The fractional Laplacian is actually a special case of this equation.
The present work will focus on these operators and in particular the associated stationary
equation, analogous to the local elliptic equation.  It should be noted that these nonlocal operators
and the associated equations can be applied not only to problems of anomalous diffusion, 
but to a wide range of phenomena, including peridynamic models of continuum mechanics which
allow crack nucleation and propagation \cite{silling2003deformation, du2012analysis}.

In this work it will be assumed that the kernel $B$ appearing above is 
parametrized by a (possibly infinite-dimensional) parameter $\lambda \in E$; $\lambda$ is assumed random.  Furthermore, some
partial noisy observations of the solution $u$ will be available.  In a probabilistic framework,
a prior distribution  is placed on $\lambda \sim \mu$ and the posterior distribution 
$\lambda|y \sim \eta$ results from conditioning on the observations.  
The joint distribution can often trivially be derived and evaluated in closed form for a given pair $(\lambda,y)$,
as $\bbP(\lambda,y) = \bbP(y|\lambda) \mu(\lambda)$.
Hence, the posterior for a given observed value of $y$ can be evaluated, up to a normalizing constant.
One aims to approximate {\it quantities of interest} $q=\int_E Q(\lambda) \nu(d\lambda)$
for some $Q:E\rightarrow \bbR$, where $\nu=\mu$ for the forward problem and $\nu=\eta$ for the inverse problem.
The likelihood $\bbP(y|\lambda)$ is often concentrated in a small, possibly nonlinear, subspace of $E$.
This posterior concentration generically precludes the naive application of standard 
forward approximation algorithms independently to the numerator, $\int_E Q(\lambda) \bbP(y|\lambda) \mu(d\lambda)$, 
and denominator, $\int_E \bbP(y|\lambda) \mu(d\lambda)$.  
It is noted that, typically, 
\eqref{eq:ell} will have to be approximated numerically, 
and this will lead to an approximate posterior density.

Monte Carlo (MC) and Sequential Monte Carlo (SMC) methods are amongst the most widely used computational techniques in statistics, engineering, physics, finance and many other disciplines.  In particular, if i.i.d.~samples $\lambda_i\sim \mu$ may be obtained, the MC sampler simply iterates this $N$ times
and approximates $(1/N) \sum_{i=1}^N Q(\lambda_i) \approx \int_E Q(\lambda)\mu(d\lambda)$.  
SMC samplers \cite{delm:06b} are designed to approximate a sequence $\{ \eta_l \}_{l \geq 0}$ of probability distributions on a common space, whose densities are only known up-to a normalising constant. The method uses $N\geq 1$ samples (or particles) that are generated in parallel, and are propagated with importance sampling (often) via Markov chain Monte Carlo (MCMC) and resampling methods. Several convergence results, as $N$ grows, have been proved (see  e.g.~\cite{delm:04}).

For problems which must first be approximated at finite resolution, as is the case in this article,
and must subsequently be sampled from using MC-based methods, 
a multilevel (ML) framework may be used. This can potentially reduce
the cost to obtain a 
given level of mean square error (MSE) \cite{heinrich2001multilevel, gile:08, giles2015multilevel},
relative to performing i.i.d.~sampling from the approximate
posterior at a given (high) resolution.
A telescopic sum of successively refined approximation
increments are estimated instead of a single highly resolved approximation. 
The convergence of the refined approximation increments allows one to balance the cost
between levels, and optimally obtain a cost $\cO(\varepsilon^{-2})$ for MSE
$\cO(\varepsilon^2)$. This is shown in the context of the models in this article.
Specifically, it is shown that the cost of MLMC, to provide a MSE of $\mathcal{O}(\varepsilon^2)$, 
is less than i.i.d.~sampling from the most accurate prior approximation.

SMC within the ML framework has been primarily developed in \cite{ourmlsmc,mlsmcnc,jasra2015multilevel}. 
These methodologies, some of which consider a similar context to this article, have been introduced where
the ML approach is thought to be beneficial, but i.i.d.~sampling is not possible. Indeed, one often has to resort to advanced MCMC or SMC methods to implement the
ML identity; see for instance \cite{hoang2013complexity} for MCMC. SMC for nonlocal problems however, is a very sensible framework to implement the ML identity, as it will approximate a sequence of related probabilities of increasing complexity; in some cases, exact simulation from couples of the probabilities is not possible.
As a result, SMC is the computational framework which we will primarily pursue in this article.
It is shown that the cost of MLSMC, to provide a MSE of $\mathcal{O}(\varepsilon^2)$, 
is less than i.i.d.~sampling from the most accurate posterior approximation.
 It is noted, however, that some developments both with regards to
theoretical analysis and implentation of the algorithm, must be carefully considered in order to successfully use MLSMC for nonlocal models.

This article is structured as follows. In section \ref{sec:theo} the nonlocal elliptic equations are given, along with the Bayesian inverse problem, and well-posedness of both are established. In particular, it is shown that the 
posterior has a well-defined Radon-Nikodym derivative w.r.t.~
 prior measure. 
In section \ref{sec:expec} we consider how one can approximate expectations w.r.t.~the prior and the posterior,
specifically using MLMC and MLSMC methods. Our complexity results are given also.
Section \ref{sec:num} provides some numerical implementations of MLMC 
and MLSMC for the prior, and posterior, respectively.

\section{Nonlocal elliptic equations}
\label{sec:theo}

\subsection{Setup}

Consider the following equation
\begin{eqnarray}
\cL(u)(x) &=& b(x), \quad {\rm for ~} x \in \Omega \subset \bbR^d \\
u(x)  &=& 0, \quad {\rm for ~} x \in \Gamma,
\label{eq:original}
\end{eqnarray}
where $\cL$ is given by \eqref{eq:ell}, 
the domain $\Omega$ is simply connected, and its ``boundary'' $\Gamma$ is sufficiently regular 
and {\it nonlocal}, in the sense that it has non-zero volume in $\bbR^d$,


Under appropriate conditions on $B$, $\cL^{-1} : H^{-s}(\Omega\cup\Gamma) \rightarrow H^{s}_c(\Omega\cup\Gamma)$, 
for $s \in [0,1)$, with $s=1$ only for the limiting local version in which 
$B(x,x') = (I - \partial^2/\partial x^2) \delta(x'-x)$ (or a similar uniformly elliptic form) \cite{du2012analysis}.  
Following \cite{du2012analysis} 
$H^s_c=\{u \in H^s; u|_\Gamma =0, \int_{\bbR^d} u(x) dx = 0\}$ denotes the volume constrained space of functions, 
and the fractional Sobolev space $H^s$ is defined as follows, for $s\in(0,1)$,
\begin{equation}
H^s := \{ u\in L^2(\Omega\cup \Gamma) : \|u\|_{L^2(\Omega\cup \Gamma)} + |u|_{H^s(\Omega\cup\Gamma) } <\infty\},
\label{eq:hs}
\end{equation}
where 
$$
|u|_{H^s(\Omega\cup\Gamma)} := \int_{\Omega\cup\Gamma} \int_{\Omega\cup \Gamma} \frac{(u(x)-u(y))^2}{|x-y|^{d+2s}} dydx.
$$
We define $H^0_c:=L^2_c$ for ease of notation.

For $b\in V^*(\Omega\cup\Gamma)$, the dual with respect to $L^2$
of some space $V(\Omega\cup\Gamma)$, the weak formulation of 
\eqref{eq:original} can be defined as follows.
Integrate the equation against an arbitrary test function $v \in V(\Omega\cup\Gamma)$.
Now, find $u \in V(\Omega\cup\Gamma)$ (satisfying the boundary conditions) 
such that
\begin{eqnarray}
\cB(u,v) & = F(v), \quad {\rm for ~ all ~ } v \in V(\Omega\cup\Gamma),
\label{eq:weak}
\end{eqnarray}
where 
\[
\begin{split}
\cB(u,v) & := \int_{\bbR^d} \int_{\bbR^d} [u(x') - u(x)] B(x,x') dx' v(x) dx, \\ 
F(v) & :=  \int_{\bbR^d} b(x) v(x) dx.
\end{split}
\]

\subsection{Numerical methods for forward solution}

Let $h_\ell$ denote the maximum diameter of an element of $V_\ell \subset V_{\ell+1} \subset \cdots \subset V$.
The finite approximation of \eqref{eq:weak} is stated as follows.  Identify some $u_\ell \in V_\ell$ such that 
\begin{eqnarray}
\cB(u_\ell,v_\ell) & = F(v_\ell), \quad {\rm for ~ all ~ } v_\ell \in V_\ell.
\label{eq:weakell}
\end{eqnarray}
Assuming the spaces $V_\ell$ are spanned by elements $\{\phi^k_\ell\}_{k=1}^{M_\ell}$, then one can substitute the 
ansatz $u_\ell = \sum_{k=1}^{M_\ell} u_{\ell}^k \phi^k_\ell$ into the above equation, resulting in the finite linear system
\begin{eqnarray}
\sum_{k=1}^{M_\ell} \cB(\phi_\ell^j,\phi^k_\ell)u_k & = F(\phi_\ell^j), \quad {\rm for ~ } j=1,\dots, M_\ell.
\label{eq:weakL}
\end{eqnarray}

In particular, the spaces $\{V_\ell\}$ will be comprised of discontinuous elements, so that the 
method described is a discontinuous Galerkin finite element method (FEM) \cite{chen2011continuous}.
Piecewise polynomial discontinuous element spaces are dense in $H^s_c$ for $s\in[0,1/2)$,
and are therefore {\it conforming} when 
$u\in H_c^s=V$ for $s\in[0,1/2)$, so there is no need to impose penalty terms at the boundaries as 
one must do for smoother problems in which the discontinuous elements are non-conforming.


\subsection{Forward UQ}

The following assumption will be made for simplicity
\begin{assumption} $V=L^2(\Omega \cup \Gamma)$, with norm $\|\cdot\|$ 
and inner product $\langle\cdot,\cdot\rangle$.  $B(x,x') = B(|x'-x|) \geq 0$ is continuous 
with $B(0)\geq c' > 0$, and there exist $K_1>0$ 
such that for all $x\in\Omega$
\begin{eqnarray}
\int_{\Omega \cup \Gamma} B(x,x') dx' & \leq K_1.  
\end{eqnarray}
\label{as:kernel}
\end{assumption}

As shown in Section 6 of \cite{gunzburger2010nonlocal}, this implies that
for all $u,v,b \in V$
\begin{itemize}
\item $F$ is continuous: $|F(v)| \leq \|b\| \|v\|$;
\item $\cB$ is continuous: $|\cB(u,v)| \leq 4K_1 \|u\|\|v\|$;
\item $\cB$ is coercive:  $|\cB(u,u)| \geq K_2 \|u\|^2$ for some $K_2(\cB,\Omega)>0$.  
This 
actually follows from the Poincar{\'e}-type inequality derived in Proposition 2.5 of \cite{andreu2009nonlocal}.
\end{itemize}

Hence, Lax-Milgram lemma (\cite{ciarlet2002finite}, Thm. 1.1.3) can be invoked, 
guaranteeing existence of a unique solution $u\in L^2(\Omega)$
such that 
\[
\|u\| \leq K_2^{-1} \|b\|.
\]
In other words, the map $b \mapsto u$ is continuous.  
The system \eqref{eq:weakL} inherits solvability since the bilinear is a fortiori coercive on $V_\ell$,
so that 
\[
\|u_\ell\| \leq K_2^{-1} \|b\|.
\]

Let $B_\lambda$ be a parametrization of $B$, where 
$\lambda \sim \mu_0$ and $\lambda \in E$, and either $E=\bbR^p$ or
$E\subset \bbR^p$ compact.  
Then the following theorem holds.  

\begin{theorem}[Well-posedness of forward UQ]  If Assumption \ref{as:kernel} holds almost surely for 
$\lambda\sim\mu$ and the map $u \mapsto Q$ is continuous from $L^2$ to $\bbR$, 
then the map $\lambda \mapsto Q$ is almost surely continuous.  
Hence $Q(\lambda) \in L^\infty \supset L^p$ for all $p\geq 1$, i.e. all moments exist. 
$L^p$ here denotes the space of random variables $X$ such that 
$\int_{E} |X|^p \mu(d\lambda) < \infty$.
\label{thm:fwd}
\end{theorem}
\begin{proof}
Since Assumption \ref{as:kernel} holds almost surely for 
$\lambda\sim\mu$, then $\| u\| \leq u^* := K_2^{-1} \|b\|$ uniformly.
So, the quantity of interest $Q(\lambda) = \|u(\lambda)\|\in L^\infty \supset L^p$ for all $p\geq 1$. 
Therefore, for any $Q: L^2(\Omega) \rightarrow \bbR$ such that 
$Q(\lambda) := Q(u(\lambda)) \leq K \|u\|$, the result follows, since then $Q(\lambda) \leq Q^* := K u^*$ for all $\lambda$.
\end{proof}

\begin{corollary}[Well-posedness of finite approximation]
If Assumption \ref{as:kernel} holds almost surely for 
$\lambda\sim\mu$ and the map $u_\ell \mapsto Q_\ell$ is continuous from $L^2$ to $\bbR$, 
then the map through the \emph{discrete} system $\lambda \mapsto Q_\ell$ is almost surely continuous
and $Q_\ell\in L^\infty$.
\label{thm:finfwd}
\end{corollary}

\subsection{Inverse UQ}

For the inverse problem, let us assume that some data is given in the form
\begin{equation}
y = \cG(\lambda) + \xi, \quad \lambda \perp \xi \sim N(0,\Sigma),
\label{eq:data}
\end{equation}
where 
\begin{equation}
\cG(\lambda) := [ \langle g_1, u(\lambda) \rangle, \dots, \langle g_M, u(\lambda) \rangle]^\top, \quad g_i\in V. 
\label{eq:g}
\end{equation}
Then the following theorem holds.

\begin{theorem}[Well-posedness of inverse UQ] 
The posterior distribution of $\lambda | y$ is well-defined and takes the form
\begin{equation}
\frac{d\eta^y}{d\mu} (\lambda) = \frac{1}{Z} \exp\{ -\frac12 | \Sigma^{-1/2}(\cG(\lambda)-y)|^2\},
\label{eq:post}
\end{equation}
with $Z = \int_E \exp\{ -\frac12 | \Sigma^{-1/2}(\cG(\lambda)-y)|^2 \} \mu(d\lambda)$.
\label{thm:inv}
\end{theorem}
\begin{proof}
The form of the posterior is obtained by a change of variables to $\{\lambda, \xi = y - \cG(\lambda)$\},
which are independent by assumption.  Note the change of variables has Jacobian 1.  
So, changing variables back, this also gives the joint density.  
The posterior is obtained by normalizing for the observed value of $y$.
The form of the observation operator guarantees $Z > \exp\{ - |\Sigma^{-1}| (|\cG^*|^2 + |y|^2)  \}$,
where $\cG^* = (\cG_1^*,\cG_2^*, \dots, \cG_M^*)$ and $\cG_i^*$ is defined as in the proof of Theorem \ref{thm:fwd},
and $|\Sigma^{-1}|$ is the matrix norm. 
\end{proof}

Now define $\cG_\ell(\lambda) := [ \langle g_1, u_\ell(\lambda) \rangle, \dots, \langle g_M, u_\ell(\lambda) \rangle], \quad g_i\in L^2(\Omega).$

\begin{corollary}[Well-posedness of inverse UQ for finite problem] 
The finite approximation of the posterior distribution of $\lambda | y$ is well-defined and takes the form
\begin{equation}
\frac{d\eta^y_\ell}{d\mu} (\lambda) = \frac{1}{Z} \exp\{ -\frac12 | \Sigma^{-1/2}(\cG_\ell(\lambda)-y)|^2\},
\label{eq:finpost}
\end{equation}
with $Z_\ell = \int_E \exp\{ -\frac12 | \Sigma^{-1/2}(\cG_\ell(\lambda)-y)|^2 \} \mu(d\lambda)$.
\label{cor:fininv}
\end{corollary}

\begin{remark}
The forcing may also be taken as uncertain, although the uniformity will require 
it to be defined on a compact space.  The probability space $E$ will be taken as compact for simplicity,
as it is then easy to verify Assumption \ref{as:kernel}. 
\end{remark}

\section{Approximation of expectations}
\label{sec:expec}

The objective here is to approximate expectations of some functional $\varphi:E\rightarrow \bbR$
with respect to a probability measure $\eta$ ($\mu$ or $\eta^y$), denoted 
$\bbE_\eta(\varphi).$ 
The solution $u$ of \eqref{eq:original} above
must be approximated by some $u_\ell$,
with a degree of
accuracy which depends on $\ell$.  
Indeed 
there exists a hierarchy of levels $\ell = 0, \dots ,L$
(where $L$ may be arbitrarily large) of increasing accuracy and increasing cost.   
For the inverse problem this manifests in a hierarchy of target probability measures 
$\{\eta_\ell\}_{\ell=0}^L$ via the approximation of \eqref{eq:post}.
For the forward problem $\eta_\ell=\eta$ for all $\ell$, 
but it will be assumed that the function $\varphi$ requires evaluation of
the solution of \eqref{eq:original}, as in Theorem \ref{thm:fwd},
and the corresponding approximations will be denoted by $\{\varphi_\ell\}_{\ell=0}^L$.
One may then compute the estimator 
$\hat{Y}_L^N = \frac1N \sum_{n=1}^N \varphi_L(\lambda_L^n), \lambda_L^n \sim \eta_L$, 
where for the forward problem it may be that $\eta_L\equiv\eta$ for all $L$.
The mean square error (MSE) is given by
\begin{equation}
\bbE \left | \bbE_{\eta}[\varphi(\lambda)] - \hat{Y}_L^N \right |^2 
=  \underbrace{\bbE\left\{ \bbE_{\eta_L}[\varphi_L(\lambda)] - \hat{Y}_L^N \right\}^2}_{\rm variance}  
+ \underbrace{\{\bbE_{\eta_L} [\varphi_L(\lambda)] - \bbE_{\eta} [\varphi(\lambda)]\}^2}_{\rm bias}\  .
\label{eq:mseml}
\end{equation}
Now assume 
there is some discretization level, say of diameter $h_L$, which gives rise to an error estimate on the output of 
size $\cO(h_L^\alpha)$, for example arising from the deterministic numerical approximation of a spatio-temporal
problem.  
This also translates to a number of degrees of freedom proportional to $h_L^{-d}$, where $d$ is the spatio-temporal dimension.
Now the complexity of typical forward solves will range from a dot product $\cO(h_L^{-d})$ (linear) 
to a full Gaussian elimination $\cO(h_L^{-3d})$ (cubic).  In this problem 
one aims to find $h_L^\alpha =\cO( \varepsilon)$, so one would find the cost controlled by $\cO(\varepsilon^{-\zeta/\alpha})$,
for $\zeta \in (d, 3d)$.  
If one can obtain independent, identically distributed samples $u^n_L$, then the necessary number of samples
to obtain a variance of size $\cO(\varepsilon^2)$ is given by $N = \cO(\varepsilon^{-2})$.
The total cost to obtain a mean-square error tolerance of $\cO(\varepsilon^2)$ is therefore $\cO(\varepsilon^{-2 - \zeta/\alpha})$.

\subsection{Multilevel Monte Carlo}
\label{sec:mlmc}

For the forward UQ problem, in which one can sample {\it directly} from $\eta_\ell$, one
very popular methodology for improving the efficiency of solution to such problems is the 
multilevel Monte Carlo (MLMC) method \cite{heinrich2001multilevel, gile:08}.
Indeed there has been an explosion of recent activity \cite{giles2015multilevel} since its introduction in \cite{gile:08}.
In this methodology the simple estimator $\hat{Y}_L^N$ above for a given desired $L$ is replaced
by a telescopic sum of unbiased  increment estimators 
$Y^{N_\ell}_\ell =  \sum_{i=1}^{N_\ell} \{ \varphi_\ell(\lambda_\ell^{(i)})-\varphi_{\ell-1}(\lambda_{\ell-1}^{(i)})\}N_\ell^{-1}$
where $\{\lambda_{\ell-1}^{(i)},\lambda_\ell^{(i)}\}$ are i.i.d.\@ samples,
with marginal laws $\eta_{\ell-1}$, $\eta_\ell$, respectively, carefully constructed on a joint
probability space.  This is repeated independently for $0\le \ell\le  L$.
The overall multilevel estimator will be 
\begin{equation}
\label{eq:multi}
\hat{Y}_{L,{\rm Multi}} = \sum_{\ell=0}^{L} Y^{N_\ell}_\ell\, ,
\end{equation}
under the convention that $g(\lambda_{-1}^{(i)})=0$.
A simple error analysis shows that 
the mean squared error (MSE) in this case is given by
\begin{eqnarray}
\nonumber
\bbE \{ \hat{Y}_{L,{\rm Multi}} - \bbE_{\eta} [\varphi(\lambda)] \}^2 & =  
\underbrace{ \sum_{\ell=0}^L \bbE \left \{ Y^{N_\ell}_{\ell} - [\bbE_{\eta_{\ell}}{\varphi_\ell}(\lambda) - \bbE_{\eta_{\ell-1}}{\varphi_{\ell-1}}(\lambda)] \right\}^2}_{\rm variance} 
\\ 
& + \underbrace{\{\bbE_{\eta_L} [{\varphi_L}(\lambda)] - \bbE_{\eta} [\varphi(\lambda)]\}^2}_{\rm bias}\  .
\label{eq:msemls}
\end{eqnarray}
Notice that the bias is given by the {\it finest} level, whilst the variance is decomposed into a sum of variances of the {\it increments}.
The variance of the $\ell^{th}$ increment estimator has the form $\cV_\ell N_\ell^{-1}$, where 
the terms $\cV_\ell=\bbE|\varphi_\ell(\lambda_\ell) - \varphi_{\ell-1}(\lambda_{\ell-1})|^2$ decay, 
following from refinement of the approximation of $\eta$ and/or $\varphi$. 
One can therefore balance the variance at a given level $\cV_\ell$ with the number of samples $N_\ell$.
As the level increases, the corresponding cost increases, but the variance {\it decreases}, allowing fewer samples to achieve a
given variance.  This can be optimized, and results in a total cost of $\cO(\varepsilon^{- \max\{2,\zeta/\alpha\}})$, in the optimal case.

To be explicit, denote by $Q_\ell := Q(u_\ell(\lambda_\ell),\lambda_\ell)$ the level $\ell$ approximation of the 
quantity of interest $Q$.
Introduce the following assumptions
\begin{hypA}
\label{hyp:D}
There exist $\alpha, \beta, \zeta>0$, 
and a $C>0$ such that 
\begin{equation}\label{eq:mlratesvan}
\begin{cases}
|\bbE(Q_L - Q_\infty) | & \leq C h_L^\alpha ;\\
\bbE|Q_\ell - Q_{\ell-1} |^2 & \leq C h_\ell^\beta ;\\              
{\rm C} (Q_{\ell}) & \leq C h_\ell^{-\zeta},
\end{cases}
\end{equation}
where ${\rm C} (Q_{\ell})$ denotes the cost to evaluate $Q_{\ell}$.
\end{hypA}

We have the following classical MLMC Theorem \cite{giles2015multilevel}

\begin{theorem}\label{theo:mlmc}
Assume (A\ref{hyp:D})
and $\rm{max}\{\beta,\zeta\} \leq 2\alpha$. 
Then for any $\varepsilon>0$, there exist $L,\{N_\ell\}_{\ell=0}^L$
and $C>0$ such that 
\begin{equation}
\mathbb{E}\Big[\Big(\hat{Y}_{L,{\rm Multi}} - \mathbb{E}_{\eta}[Q]\Big)^2\Big] \leq C \varepsilon^2,
\label{eq:mlmse}
\end{equation}
for the following cost 
\begin{equation}\label{eq:mlCostsvan}
{\rm COST} \leq C 
\begin{cases}
\varepsilon^{-2}, & \text{if} \quad \beta > \zeta, \\ 
\varepsilon^{-2} |\log(\varepsilon)|^2, & \text{if} \quad \beta = \zeta, \\ 
\varepsilon^{-\left( 2 + \frac{\zeta-\beta}{\alpha} \right)}, & \text{if} \quad \beta < \zeta. 
\end{cases}
\end{equation}
\end{theorem}

\subsection{Multilevel sequential Monte Carlo sampler}
\label{sec:mlsmc}

Now the inverse problem will be considered.
There is a sequence of probability measures $\{\eta_{\ell}\}_{\ell\geq 0}$ on a common measurable 
space $(E,\mathcal{E})$,
and for each $l$ there is a measure
$\gamma_{\ell}:\mathcal{E}\rightarrow\mathbb{R}^+$, 
such that
%
\begin{equation}
\label{eq:target}
\eta_{\ell}(d\lambda) = \frac{\gamma_\ell(d\lambda)}{Z_\ell}
\end{equation}
where the normalizing constant $Z_\ell = \int_E\gamma_\ell(d\lambda)$ 
is unknown.  
The objective is to compute:
$$
\mathbb{E}_{\eta_\infty}[\varphi(\lambda)] := \int_E \varphi(\lambda)\eta_\infty(d\lambda) 
$$
for potentially many measurable $\eta_\infty-$integrable functions $\varphi:E\rightarrow\mathbb{R}$. 

\subsubsection{Notations}

Let $(E,\mathcal{E})$ be a measurable space.
The notation $\testfunction(E)$ denotes the class of bounded and measurable real-valued functions, and the 
supremum norm is written as $\|f\|_{\infty} = \sup_{\lambda\in E}|f(\lambda)|$. 
Consider non-negative operators 
$K : E \times \mathcal{E} \rightarrow \bbR_+$ such that for each $\lambda \in E$ the mapping $A \mapsto K(\lambda, A)$ is a finite non-negative measure on $\mathcal{E}$ and for each $A \in \mathcal{E}$ the function $\lambda \mapsto K(\lambda, A)$ is measurable; the kernel $K$ is Markovian if $K(\lambda, dv)$ is a probability measure for every $\lambda \in E$.
For a finite measure $\mu$ on $(E,\mathcal{E})$,  and a real-valued, measurable $f:E\rightarrow\mathbb{R}$, we define the operations:
\begin{equation*}
    \mu K  : A \mapsto \int K(\lambda, A) \, \mu(d\lambda)\ ;\quad 
    K f :  \lambda \mapsto \int f(v) \, K(\lambda, dv).
\end{equation*}
We also write $\mu(f) = \int f(\lambda) \mu(d\lambda)$. In addition $\|\cdot\|_{r}$, $r\geq 1$, denotes the $L_r-$norm, where the expectation is w.r.t.~the law of the appropriate simulated algorithm.

\subsubsection{Algorithm}

As described in Section \ref{sec:intro}, the context of interest is when a sequence of densities 
$\{\eta_{\ell}\}_{\ell\ge 0}$, as in (\ref{eq:target}), are
associated to an `accuracy' parameter $h_l$, with $h_\ell\rightarrow 0$ 
as $\ell\rightarrow \infty$, such that $\infty>h_0>h_1\cdots>h_{\infty}=0$.  
In practice one cannot treat $h_\infty=0$ and so must consider these distributions with $h_\ell>0$.
The laws with large $h_\ell$ are easy to sample from with low computational cost, but are very different from $\eta_{\infty}$, whereas, those distributions with small $h_\ell$ are
hard to sample with relatively high computational cost, but are closer to $\eta_{\infty}$. 
Thus, we choose a maximum level $L\ge 1$ and we will estimate
$$
\mathbb{E}_{\eta_L}[\varphi(\lambda)] := \int_E \varphi(\lambda)\eta_L(\lambda)d\lambda\ .
$$
By the standard telescoping identity used in MLMC, one has
\begin{align}
\mathbb{E}_{\eta_L}[\varphi(\lambda)] & =  \mathbb{E}_{\eta_0}[\varphi(\lambda)] + \sum_{\ell=1}^{L}\Big\{
\mathbb{E}_{\eta_\ell}[\varphi(\lambda)] - \mathbb{E}_{\eta_{\ell-1}}[\varphi(\lambda)]\Big\} \nonumber \nonumber \\ 
& =\mathbb{E}_{\eta_0}[\varphi(\lambda)] + \sum_{\ell=1}^{L}\mathbb{E}_{\eta_{\ell-1}}\Big[
\Big(\frac{\gamma_\ell(\lambda)Z_{\ell-1}}{\gamma_{\ell-1}(\lambda)Z_\ell} - 1\Big)\varphi(\lambda)\Big]\ .
\label{eq:ml_approx}
\end{align}

Suppose now that one applies an SMC sampler \cite{delm:06b} to obtain 
a collection of samples (particles) that sequentially approximate $\eta_0, \eta_1,\ldots, \eta_L$. 
We consider the case when one initializes the population of particles by sampling  i.i.d.~from $\eta_0$, then at every step  resamples and applies a MCMC kernel to mutate the particles.
We denote by $(\lambda_{0}^{1:N_0},\dots,\lambda_{L-1}^{1:N_{L-1}})$, with $+\infty > N_0\geq N_1\geq \cdots  N_{L-1}\geq 1$, the samples after mutation; one resamples $\lambda_l^{1:N_l}$ according to the weights $G_{\ell}(\lambda_\ell^i) = 
(\gamma_{\ell+1}/\gamma_\ell)(\lambda_\ell^{i})$, for indices $l\in\{0,\dots,L-1\}$.
We will denote by $\{M_\ell\}_{1\leq \ell \leq L-1}$ the sequence of MCMC kernels used at stages $1,\dots,L-1$, such that $\eta_{\ell}M_\ell = \eta_\ell$.
For $\varphi:E\rightarrow\mathbb{R}$, $\ell \in\{1,\dots,L\}$, we have the following estimator 
of $\bbE_{\eta_{\ell-1}}[\varphi(\lambda)]$:
$$
\eta_{\ell-1}^{N_{\ell-1}}(\varphi) = \frac{1}{N_{\ell-1}}\sum_{i=1}^{N_{\ell-1}}\varphi(\lambda_{\ell-1}^i)\ . 
$$
We define
$$
\eta_{\ell-1}^{N_{\ell-1}}(G_{\ell-1}M_\ell(d\lambda_\ell)) = \frac{1}{N_{\ell-1}}\sum_{i=1}^{N_{\ell-1}}G_{\ell-1}(\lambda_{\ell-1}^i) M_\ell(\lambda_{\ell-1}^i,d\lambda_\ell)\ .
$$
The joint probability distribution for the SMC algorithm is 
$$
\prod_{i=1}^{N_0} \eta_0(d\lambda_0^i) \prod_{\ell=1}^{L-1} 
\prod_{i=1}^{N_\ell} \frac{\eta_{\ell-1}^{N_{\ell-1}}(G_{\ell-1}M_\ell(d\lambda_\ell^i))}{\eta_{\ell-1}^{N_{\ell-1}}(G_{\ell-1})}\ .
$$
If one considers one more step in the above procedure, that would deliver samples 
$\{\lambda_L^i\}_{i=1}^{N_L}$, a standard SMC sampler estimate of the quantity of interest in (\ref{eq:ml_approx})
is $\eta_L^{N}(g)$; the earlier samples are discarded. 
Within a multilevel context, a consistent SMC estimate of \eqref{eq:ml_approx}
is
\begin{equation}
\widehat{Y} =
\eta_{0}^{N_0}(\varphi) + \sum_{\ell=1}^{L}\Big\{\frac{\eta_{\ell-1}^{N_{\ell-1}}(\varphi G_{\ell-1})}{\eta_{\ell-1}^{N_{\ell-1}}(G_{\ell-1})} - \eta_{\ell-1}^{N_{\ell-1}}(\varphi)\Big\}\label{eq:smc_est}\ , 
\end{equation}
and this will be proven to be superior than the standard one, under assumptions.

The relevant MSE error decomposition here is:
\begin{equation}
\label{eq:dec}
\mathbb{E}\big[ \{\widehat{Y}-\mathbb{E}_{\eta_\infty}[\varphi(\lambda)] \}^2\big]
\le 2\,\mathbb{E}\big[\{\widehat{Y}-\mathbb{E}_{\eta_L}[\varphi(\lambda)]\}^2\big] + 
2\,\{ \mathbb{E}_{\eta_L}[\varphi(\lambda)] - \mathbb{E}_{\eta_\infty}[\varphi(\lambda)]  \}^2\ . 
\end{equation} 

\subsubsection{Multilevel SMC}\label{sec:complex}

We will now 
restate an analytical result from \cite{ourmlsmc} that controls the error term 
$\mathbb{E}[\{\widehat{Y}-\mathbb{E}_{\eta_L}[\varphi(\lambda)]\}^2]$ in expression \eqref{eq:dec}.
For any $\ell\in\{0,\dots,L\}$ and $\varphi\in \mathcal{B}_b(E)$ we write:
$
\eta_\ell(\varphi) := \int_E \varphi(\lambda)\eta_\ell(\lambda)d\lambda.
$
The following standard assumptions will be made ; see \cite{ourmlsmc, delm:04}.


\begin{hypA}
\label{hyp:A}
There exist $0<\underline{C}<\overline{C}<+\infty$ such that
\begin{eqnarray*}
\sup_{\ell \geq 1} 
\sup_{\lambda\in E} G_\ell (\lambda) & \leq & \overline{C}\ ;\\
\inf_{\ell \geq 1} 
\inf_{\lambda\in E} G_\ell (\lambda) & \geq & \underline{C}\ .
\end{eqnarray*}
\end{hypA}

\begin{hypA}
\label{hyp:B}
There exists a $\rho\in(0,1)$ such that for any $\ell\ge 1$, $(\lambda,v)\in E^2$, $A\in\mathcal{E}$:
$$
\int_A M_\ell(\lambda,d\lambda') \geq \rho \int_A M_\ell(v,d\lambda')\ .
$$
\end{hypA}

Under these assumptions the following Theorem is proven in \cite{ourmlsmc}

\begin{theorem}\label{theo:main_error}
Assume (A\ref{hyp:A}-\ref{hyp:B}). There exist $C<+\infty$ 
such  that for any $\varphi \in\mathcal{B}_b(E)$, with $\|\varphi\|_{\infty}=1$,
\begin{align*}
\mathbb{E}\big[\{\widehat{Y}-\mathbb{E}_{\eta_L}[\varphi(\lambda)]\}^2\big] 
\leq 
C\,\bigg(\frac{1}{N_0} + &\sum_{\ell=1}^{L}\bigg( \tfrac{\cV_l}{N_{l-1}} 
+ 
\tfrac{\cV_\ell^{1/2}}{N_{\ell-1}^{1/2}} \sum_{q=\ell+1}^{L}\tfrac{\cV_q^{1/2}}{N_{q-1}} \bigg)\bigg)\ ,
\end{align*}
where $\cV_\ell := \|\tfrac{Z_{\ell-1}}{Z_{\ell}}G_{\ell-1}-1\|_{\infty}^2$.

\end{theorem}


The following additional assumption will now be made

\begin{hypA}
\label{hyp:C}
There exist $\alpha, \beta, \zeta>0$, 
and a $C>0$ such that 
\begin{equation}\label{eq:mlrates}
\begin{cases}
\cV_\ell & \leq C h_\ell^\beta ;\\              
|(\eta_L - \eta_\infty)(1) | & \leq C h_L^\alpha ;\\
{\rm C} (G_{\ell}) & \leq C h_\ell^{-\zeta},
\end{cases}
\end{equation}
where ${\rm C} (G_{\ell})$ denotes the cost to evaluate $G_{\ell}$.
\end{hypA}

The following Theorem may now be proven

\begin{theorem}\label{theo:mlsmc}
Assume (A\ref{hyp:A}-\ref{hyp:C})
and $\rm{max}\{\beta,\zeta\} \leq 2\alpha$. 
Then for any $\varepsilon>0$, there exist $L,\{N_\ell\}_{\ell=0}^L$
and $C>0$ such that 
\begin{equation}
\mathbb{E}\Big[\Big(\widehat{Y} - \mathbb{E}_{\eta}[\varphi(\lambda)]\Big)^2\Big] \leq C \varepsilon^2,
\label{eq:mlsmcmse}
\end{equation}
for the following cost 
\begin{equation}\label{eq:mlncCosts}
{\rm COST} \leq C 
\begin{cases}
\varepsilon^{-2}, & \text{if} \quad \beta > \zeta, \\ 
\varepsilon^{-2} |\log(\varepsilon)|^2, & \text{if} \quad \beta = \zeta, \\ 
\varepsilon^{-\left( 2 + \frac{\zeta-\beta}{\alpha} \right)}, & \text{if} \quad \beta < \zeta. 
\end{cases}
\end{equation}
\end{theorem}

\begin{proof}
The MSE can be bounded using \eqref{eq:dec}. 
Following from (A\ref{hyp:C})(ii), the second term requires that $h_L^{\alpha} \eqsim \varepsilon$, and assuming
$h_L = M^{-L}$ for some $M\geq2$, this translates to $L\eqsim \log\varepsilon$.  
As in \cite{ourmlsmc}, the additional error term 
is dealt with by first ignoring it and optimizing
${\rm COST}({\bf N})$, for a given $\cV({\bf N})=\varepsilon^2$, where ${\bf N} := (N_1,\dots, N_L)$. 
This requires that $N_\ell \propto \sqrt{\cV_l/C_l} \eqsim h_\ell^{(\beta+\zeta)/2}$.
The constraint then requires that $N_\ell=\varepsilon^{-2}K_L h_\ell^{(\beta+\zeta)/2}$, 
where $K_L = \sum_{l=1}^{L-1} h_l^{(\beta-\zeta)/2}$, so one has 
$$
{\rm COST}({\bf N}) = \sum_{\ell=0}^L N_\ell C_\ell = \varepsilon^{-2} K_L^{2}.
$$
It was shown in \cite{ourmlsmc} that the result follows, provided $\zeta<2\alpha$.
In fact, re-examining the additional term as in Section 3.3 of \cite{ourmlsmc}, one has
\[
\begin{split}
\sum_{\ell=1}^{L} \tfrac{\cV_\ell^{1/2}}{N_{\ell-1}^{1/2}} \sum_{q=\ell+1}^{L}\tfrac{\cV_q^{1/2}}{N_{q-1}}  
& = \cO(\varepsilon^{2} \varepsilon^{1-\zeta/2\alpha} \sum_{\ell=1}^{L-1} h_\ell^{(\beta-\zeta)/4} K_L^{-3/2}) \\
& = \cO(\varepsilon^{2} \varepsilon^{1-\zeta/2\alpha} L^{1/2} K_L^{-1}), 
\end{split}
\]
where the second line follows from the inequality $(\sum_{\ell=0}^L a_\ell)^2 \leq L \sum_{\ell=0}^L a_\ell^2$ and the definition of $K_L$.
Now substituting $K_L= \cO(1), \cO(L), \cO(\varepsilon^{-(\zeta-\beta)/2\alpha})$ for the 3 cases 
and recalling the assumption $\rm{max}\{\beta,\zeta\} \leq 2\alpha$ gives $\cV({\bf N})=\varepsilon^2$
for the costs in \eqref{eq:mlncCosts}.

\end{proof}



\subsubsection{Verification of assumptions}

Assume a uniform prior $\mu$. 
Following from \eqref{eq:data}, the unnormalized measures will be given by 
\begin{equation}
\gamma_\ell = \exp\{-\Phi(\cG_\ell(\lambda))\},
\label{eq:gamdef}
\end{equation}
where $\Phi(\cG_l(\lambda)) = \frac12|\Sigma^{-1/2} (\cG_l(\lambda) - y) |^2$, and
$$
\cG_\ell(\lambda) := [ \langle g_1, u_\ell(\lambda) \rangle, \dots, \langle g_M, u_\ell(\lambda) \rangle]^\top, 
\quad g_i\in L^2(\Omega), 
$$
and $u_l$ is the solution of the numerical approximation of \eqref{eq:original}.
Notice that these are uniformly bounded, over both $\lambda \in E$ and over $\ell$, following
from Corollary \ref{cor:fininv} for finite $\ell$, and Theorem \ref{thm:inv} in the limit.  

It is shown in \cite{ourmlsmc} Section 4 that 
\begin{equation}
|\Phi(\cG_\ell(\lambda)) - \Phi(\cG_{\ell-1}(\lambda))| \leq C(|\cG_\ell | , |\cG_{\ell-1} | ) |\cG_\ell(\lambda) - \cG_{\ell-1}(\lambda)|. 
\label{eq:phig}
\end{equation}
One proceeds similarly to that paper, and finds that
\begin{equation}
|\cG_\ell(\lambda) - \cG_{\ell-1}(\lambda)| 
\leq C \| u_\ell(\lambda) - u_{\ell-1}(\lambda) \|_V.
\label{eq:gu}
\end{equation}

Note that $G_\ell = \gamma_{\ell+1}/\gamma_\ell$,
so inserting the bound \eqref{eq:phig} into \eqref{eq:gamdef}, 
and observing the boundedness of the $\{\cG_\ell\}$
noted above, one can see that Assumption (A\ref{hyp:A}) holds. 
Now inserting \eqref{eq:phig} and \eqref{eq:gu} into \eqref{eq:gamdef}, 
one can see that in order to establish Assumption (A\ref{hyp:C}), 
it suffices to establish 
rates of convergence 
for $\| u_\ell(\lambda) - u_{\ell-1}(\lambda) \|_V$.  In particular, Assumption (A\ref{hyp:A})
and the C2 inequality (Minkowski plus Young's) provide that 
\[
\begin{split}
\cV_\ell  & = \|\tfrac{Z_{\ell-1}}{Z_{\ell}}G_{\ell-1}-1\|_{\infty}^2 = \left \|\frac{G_{\ell-1}}{\eta_{\ell-1}(G_{\ell-1})}-1\right \|_{\infty}^2 \\
& \leq C ( \|{G_{\ell-1}}-1\|_{\infty}^2 + \|\eta_{\ell-1}(G_{\ell-1}) - 1 \|_\infty^2 )  \leq 2C \|{G_{\ell-1}}-1\|_{\infty}^2 \\
& \leq C' \| u_\ell(\lambda) - u_{\ell-1}(\lambda) \|_V^2.
\end{split}
\]

Therefore, the rate of convergence of $\| u_\ell(\lambda) - u_{\ell-1}(\lambda) \|_V$ 
is the required quantity for both the forward (for Lipschitz $Q$) and the inverse 
multilevel estimation problems.  By the triangle inequality it suffices to consider 
the approximation of the truth $\| u_\ell(\lambda) - u(\lambda) \|_V$, and 
by C{\'e}a's lemma (\cite{ciarlet2002finite}, Thm. 2.4.1),
Assumption \ref{as:kernel} guarantees  the existence of a $C>0$ such that 
\begin{equation}
\| u_\ell(\lambda) - u(\lambda) \|_V \leq C \inf_{v_\ell \in V_\ell} \| v_\ell(\lambda) - u(\lambda) \|_V 
\label{eq:cea}
\end{equation}

Theorem 6.2 of \cite{du2012analysis} provides 
rates of convergence 
of the best approximation above, in the case that $V=H^s$ and the FEM 
triangulation is shape regular and quasi-uniform.  
In particular, their Case 1 corresponds to the case of a singular kernel, i.e. not even integrable, 
and $\cL:H^s_c \rightarrow H^{-s}$, for $s\in (0,1)$, so the solution 
operator is smoothing.
Their Case 2 corresponds to a slightly more regular kernel than ours, 
where in fact $B(x,\cdot) \in L^2$, rather than $L^1$ as given by Assumption \ref{as:kernel}
and consequently $\cL:L^2_c \rightarrow L^2$.
It is shown that in Case 2, if the solution $u\in H^{m+t}$, for polynomial elements of order $m$
and some $t\in[0,1]$, then convergence with respect to the $L^2$ norm is in fact $\cO(h^{m+t})$.  
So, for linear elements, second order convergence can be obtained, 
leading to $\beta=4$ in \eqref{eq:mlrates}.

Recall that \ref{as:kernel} 
ensures well-posedness of the solution from $ L^2 \ni b \mapsto u\in L^2$,
and so discontinuities are allowed.  This is actually a strength of nonlocal models and one impetus for their use.
It is reasonable to expect that if the the nodes of the discontinuous elements match up 
with the discontinuities, i.e. there is a node at each point of discontinuity for $d=1$ or there are nodes
all along a curve/surface of discontinuity for $d>1$, and if the solution is sufficiently smooth in the subdomains, 
then the rates of convergence should match those of the smooth subproblems.  
For example, if there is a point discontinuity such that the domain can be separated at that point into 
$\Omega_1\cup\Omega_2=\Omega$ and one has 
$u|_{\Omega_1} \in H^{m+t}(\Omega_1)$ and $u|_{\Omega_2} \in H^{m+t}(\Omega_2)$, where $u|_{\Omega_i}$ is the
restriction of $u$ to the set $\Omega_i$, then one expects the convergence rate of $m+t$ to be preserved.  
This is postulated and verified numerically in \cite{chen2011continuous}.
It is also illustrated numerically in that work that even with discontinuous elements, 
if there is no node at the discontinuity then the convergence rate reverts to $1/2$, 
so $\beta=1$ in \eqref{eq:mlrates}.


\section{Numerical examples}
\label{sec:num}

The particular nonlocal model which will be of interest here is that in which
the kernel is given by 
\begin{equation}
  B_\lambda(x,x') = f(x,x',\theta) \frac{1}{\delta^2|x - x'|^\alpha}
  {\mathbf{1}}_{\{|x - x'| < \delta\}}, \label{eq:part}
\end{equation}
where $\lambda := (\theta, \alpha, \delta)$, and ${\mathbf{1}}_{A}$ is the
characteristic function which takes the value 1 if $(x,x')\in A$ and zero
otherwise, and $f(x,x')=f(x',x)$. Notice that $\alpha\in (0,d)$ is scalar, but
$(\theta,\delta)$ may be defined on either a finite-dimensional subspace of
function-space, or in principle in the full infinite dimensional space.  This
may be of interest to incorporate spatial dependence of material properties. 

Following \cite{d2014optimal} we consider the
following model. Let $\Omega = [0, 1]$ and $\Gamma = [-\delta, 0)\cup(1, 1 +
\delta]$. For $(x,x')\in \Omega^2$, let $z = (x + x')/2$.  The kernel is
defined by,
\begin{align*}
  f(x, x',\theta) &= \begin{cases}
    0.2 + (z - 0.625)^2   & \text{if } z \in [0, 0.625) \\
    z + \theta            & \text{if } z \in [0.625, 0.75) \\
    14.4 + (z - 0.75) + 2 & \text{if } z \in [0.75, 1)
  \end{cases} \\
  b(x) &= 5
\end{align*}
The following prior is used for the parameters,
\begin{align*}
  \theta & \sim \mathrm{Uniform}(1, 2) \\
  \alpha & \sim \mathrm{Beta}(2, 2) \\
  \delta & \sim \mathrm{Gamma}(1, 1)\text{ truncated on }(0.125,1)
\end{align*}
This is an extension of an example used in \cite{d2014optimal}, which is
identical to this model with $\theta = 1.25$ and $\alpha = 1$ (various values of $\delta$ are used there). 
To solve the equation with FEM, 
a uniform partition on $\Omega$ is used with $h_\ell =
2^{-(k + \ell)}$, $k = 3$. Thus for each level $\ell$, $h < \delta$. 
The same discontinuous Galerkin FEM 
as in \cite{d2014optimal} is applied to solve the system. 

Similar methods to those in \cite{ourmlsmc} are used to estimate the convergence rates
for both the forward and inverse problems. The estimated rates are $\hat\alpha =
2.005$ and $\hat\beta = 4.271$. The rates $\alpha = 2$ and $\beta = 4$ are used
for the simulations.

For the inverse problem, data is generated with $\theta = 2$, $\alpha = 0.5$ and $\delta = 0.2$. 
The observations are $y|u \sim N(m(u),\sigma^2)$, where $\sigma^2=0.01$ and $m(u)=[u(a),u(b)]^\top$.
The quantity of interest for both the forward and inverse problem is 
$u(0.5)$.

The forward problem is solved in the standard multilevel fashion
by generating coupled particles from the prior at each level. 
The main result, the cost vs. MSE of the estimates is shown in
Figure~\ref{fig:mlmc}.  The Bayesian inverse problem is also solved for this
model, and the main result is shown in Figure~\ref{fig:mlsmc}.

\begin{figure}
  \includegraphics[width=\linewidth]{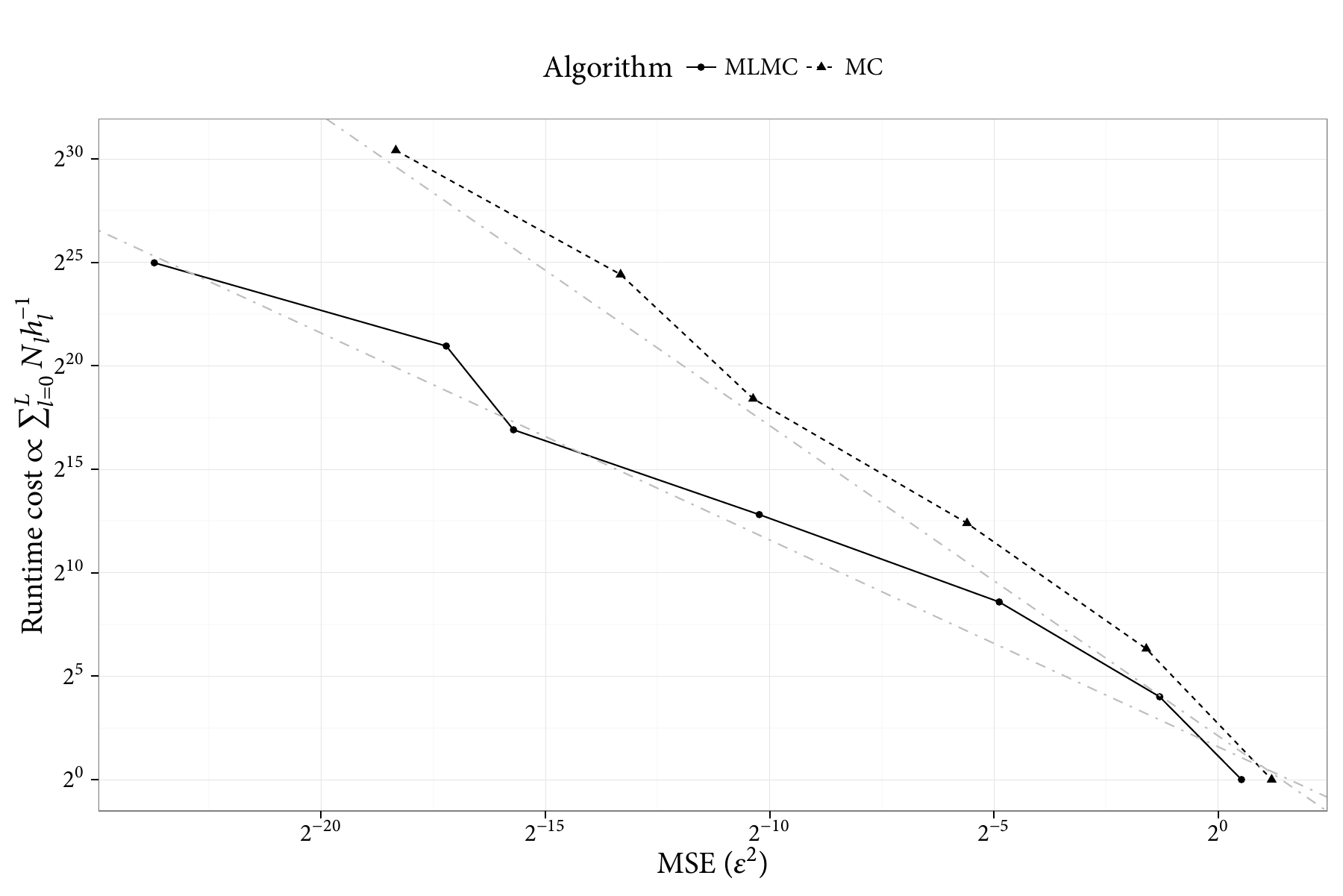}
  \caption{Cost vs. MSE for MLMC}
  \label{fig:mlmc}
\end{figure}

\begin{figure}
  \includegraphics[width=\linewidth]{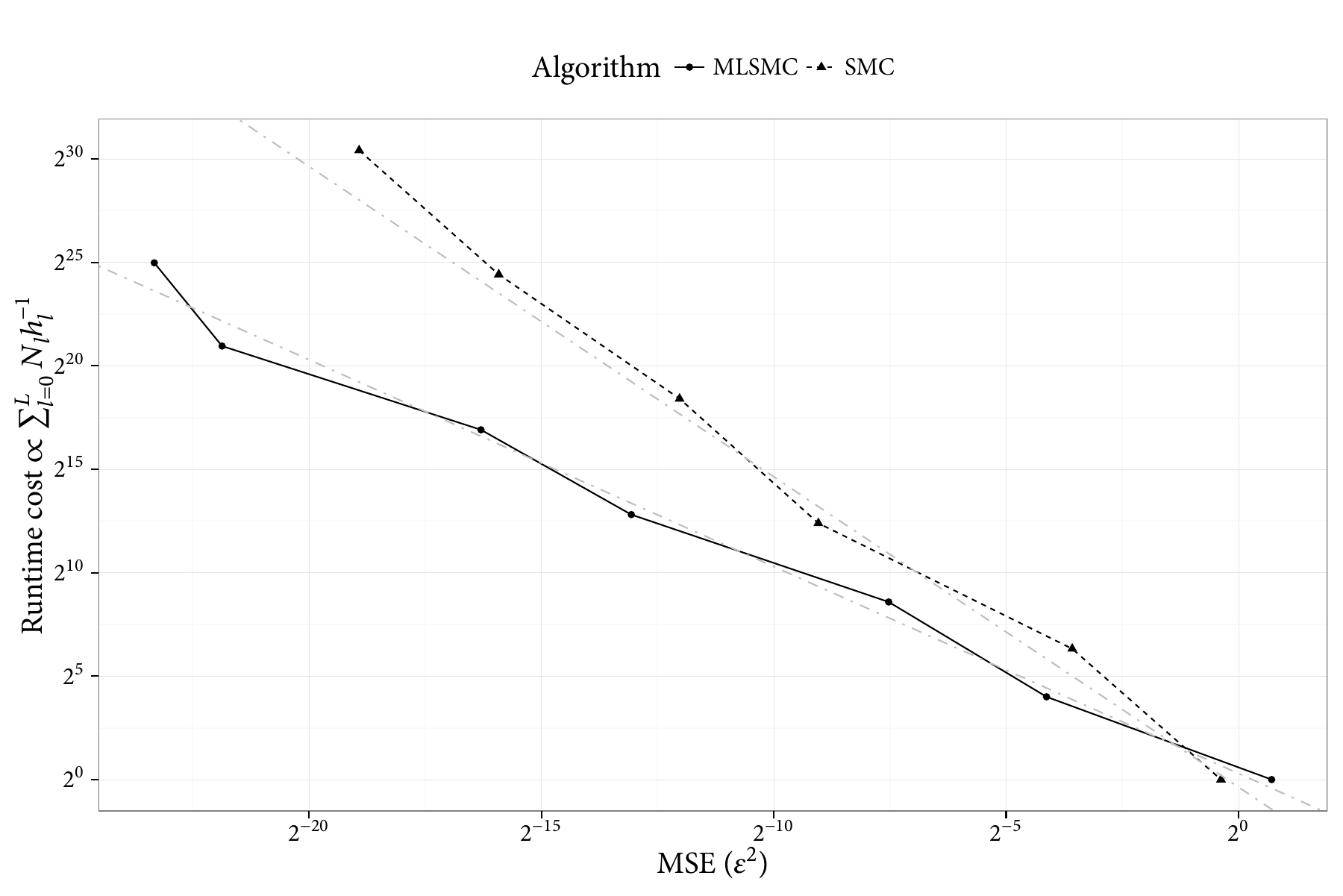}
  \caption{Cost vs. MSE for MLSMC}
  \label{fig:mlsmc}
\end{figure}

\section{Summary}

This is the first systematic treatment of UQ for nonlocal models, to the knowledge of the authors.  
Natural extensions include obtaining rigorous convergence results for piecewise smooth solutions, 
exploring higher-dimensional examples, spatial parameters, time-dependent models, 
and parameters defined on non-compact spaces.

{\bf Acknowledgements.  } KJHL was supported by the DARPA FORMULATE project. 
AJ \& YZ were supported by Ministry of Education AcRF tier 2 grant, R-155-000-161-112.
We express our gratitude to Marta D'Elia, Pablo Seleson, and Max Gunzburger for useful discussions.


\end{document}